\documentclass[journal]{IEEEtran}
\usepackage{mathrsfs}
\usepackage{amsmath}
\usepackage{amssymb}
\usepackage{mathtools} 
\usepackage{mathptmx}
\DeclareMathAlphabet{\mathcal}{OMS}{cmsy}{m}{n}
\DeclareSymbolFont{largesymbols}{OMX}{cmex}{m}{n}
\usepackage{amsthm}
\usepackage{amstext}
\usepackage{bm}
\usepackage{verbatim}
\usepackage{amsfonts}
\usepackage{cite}
\usepackage{array}
\usepackage{algorithm}
\usepackage{graphicx}
\usepackage{float}
\usepackage{lipsum}
\usepackage{textcomp}
\usepackage{gensymb}
\usepackage{multirow}
\usepackage{mathrsfs}
\usepackage{mathtools} 
\usepackage{algpseudocode}
\usepackage{booktabs}

\usepackage{threeparttable}
\usepackage{nomencl}
\usepackage{enumitem}
\newenvironment{termtable}[1][2cm]{%
	\def\term##1##2{\item[$##1$] ##2}%
	\itemize[left=0pt .. #1, itemindent=0pt,
	align=parleft, nosep]
}{%
	\enditemize
}
\makenomenclature

\usepackage{url}

\allowdisplaybreaks[2]

\usepackage{float}
\newcommand{\RNum}[1]{\uppercase\expandafter{\romannumeral #1\relax}}

\def\baa{\begin{align}}
\def\eaa{\end{align}}

\newcommand{\bsq}{\begin{subequations}}
	\newcommand{\esq}{\end{subequations}}

\newcommand{\beq}{\begin{equation}}
\newcommand{\eeq}{\end{equation}}
\newcommand{\bq}{\begin{eqnarray}}
\newcommand{\eq}{\end{eqnarray}}
\newcommand{\bqn}{\begin{eqnarray*}}
	\newcommand{\eqn}{\end{eqnarray*}}
\newcommand{\bee}{\begin{enumerate}}
	\newcommand{\eee}{\end{enumerate}}
\newcommand{\bi}{\begin{itemize}}
	\newcommand{\ei}{\end{itemize}}

\usepackage{comment}
\usepackage{color}
\newboolean{showcomments}
\setboolean{showcomments}{true}
\newcommand{\wang}[1]{\ifthenelse{\boolean{showcomments}}
	{ \textcolor[rgb]{1,0,1}{(ZW:  #1)}}{}}
\newcommand{\fliu}[1]{\ifthenelse{\boolean{showcomments}}
	{ \textcolor{red}{(FL:  #1)}}{}}
\newcommand{\zhang}[1]{\ifthenelse{\boolean{showcomments}}
	{ \textcolor{blue}{(YFZ:  #1)}}{}}

\theoremstyle{definition}
\newtheorem{theorem}{Theorem}
\newtheorem{lemma}[theorem]{Lemma}

\theoremstyle{definition}

\ifCLASSOPTIONcompsoc
\usepackage[caption=false, font=normalsize, labelfont=sf, textfont=sf]{subfig}
\else
\usepackage[caption=false, font=footnotesize]{subfig}

\begin{document}

\title{Robust Scheduling of Virtual Power Plant under Exogenous and Endogenous Uncertainties}
% 
% Day-Ahead Scheduling of Virtual Power Plant in DSO’s Energy-Reserve Pool: under Exogenous and Endogenous Uncertainties
\author{Yunfan ~Zhang,
        ~Feng ~Liu,~\IEEEmembership{Senior Member,~IEEE,}
        ~Zhaojian ~Wang,~\IEEEmembership{Member,~IEEE,}
        ~Yifan ~Su,
        ~Weisheng~Wang,
        ~and~Shuanglei~Feng
        % <-this % stops a space
%\thanks{M. Shell was with the Department of Electrical and Computer Engineering, Georgia Institute of Technology, Atlanta, GA, 30332 USA e-mail: (see http://www.michaelshell.org/contact.html).}% <-this % stops a space
}
\markboth{Journal of \LaTeX\ Class Files,~Vol.~xx, No.~xx, August~xxxx}%
{Shell \MakeLowercase{\textit{et al.}}: Bare Demo of IEEEtran.cls for IEEE Journals}

\maketitle

\begin{abstract}
Virtual power plant (VPP) provides a flexible solution to distributed energy resources integration by  aggregating renewable generation units, conventional power plants, energy storages, and flexible demands. This paper proposes a novel model for determining the optimal offering strategy in the day-ahead energy-reserve market and the optimal self-scheduling plan. It considers exogenous uncertainties (or called decision-independent uncertainties, DIUs)  associated with market clearing prices and available wind power generation, as well as the endogenous uncertainties (or called decision-dependent uncertainties, DDUs)  pertaining to real-time reserve deployment requests. A tractable solution method based on strong duality theory, McCormick relaxation, and the Benders' decomposition to solve the proposed stochastic adaptive robust optimization with DDUs formulation is developed. Simulation results demonstrate the applicability of the proposed approach.
\end{abstract}

\begin{IEEEkeywords}
Adaptive robust optimization, decision dependent uncertainty, endogenous uncertainty, virtual power plant
\end{IEEEkeywords}

\IEEEpeerreviewmaketitle
\section*{Notation}
In this paper, $\mathbb{R}^n$($\mathbb{R}^{m\times n}$) depicts the $n$-dimensional Euclidean space. $[n]:=\left\{1,,...,n\right\}$ denotes the set of integers from 1 to $n$. For a column vector ${x}\in\mathbb{R}^n$ (matrix $A\in\mathbb{R}^{m\times n}$), ${x}^{\mathsf{T}}$ ($A^{\mathsf{T}}$) denotes its transpose. We use ${1}$ and ${0}$ to denote vector of ones and zeros, respectively. For ${x},{ y}\in \mathbb{R}^n$, we denote the inner product ${ x}^{\mathsf{T}}{y}=\sum_{i=1}^n x_iy_i$ where $x_i,y_i$ stands for the $i$-th entry of ${x}$ and ${ y}$, respectively. We use flourish capital $\mathcal{W}:{X}\rightrightarrows{Y}$ to denote a set-valued map if $\mathcal{W}({x})$ is a nonempty subset of ${Y}$ for all ${ x}\in{X}$.
%Additional symbols with subscript $k$ are used to represent new variables corresponding to iteration k. 

\subsection*{Sets and Index}
\begin{termtable}[2.5cm]
	\term{T}{Sets of time periods indexes $t$.}
	\term{I_G,I_W,I_D,I_{ES}}{Sets of conventional power plants (CPP), wind generation units, flexible demand units, and energy storage units indexes $i$.}
	\term{I_{S}}{Sets of market price scenarios indexes $s$.}
\end{termtable}
\subsection*{Parameters}
\begin{termtable}[2.5cm]
	\term{C_i^{0},C_i^{1},C_i^{SU},C_i^{SD}}{Fixed, variable, start-up and shut-down cost coefficient of CPP $i$.}
	\term{\overline{P}^{R+},\overline{P}^{R-}}{Maximal up-/down- reserve capacity.}
	\term{\overline{E}^{R+},\overline{E}^{R-}}{Maximal up-/down- reserve energy that can be traded in the reserve markets.}
	\term{EXCH_{\rm max}}{Transaction limit between the VPP and distribution energy market.}
	\term{T_i^{\rm on},T_i^{\rm off}}{Minimal on/off time of CPP $i$.}
	\term{R_i^{+},R_i^{-},R_i^{SU},R_i^{SD}}{Up-/down-/ start-up/shut-down ramping limits of the CPP $i$.}
	\term{\overline{P}_i^{Gen},\underline{P}_i^{Gen}}{Power limit of the CPP $i$.}
	\term{\overline{P}_{i,t}^D,\underline{P}_{i,t}^D,\underline{D}_i^D}{Maximal/minimal power consumption and minimal daily energy consumption of the flexible demand $i$.}
	\term{r_i^{D+},r_i^{D-}}{Load pick-up/drop ramping limits for flexible demand.}
	\term{\overline{P}_i^{ch},\overline{P}_i^{dc}}{Charging and discharging power capacities of the storage unit $i$.}
	\term{\eta_i^{ch},\eta_i^{dc}}{Charging and discharging efficiency rates of the storage unit $i$.}
	\term{\overline{SOC}_i,\underline{SOC}_i}{Lower and upper bounds for the stored energy of the storage unit $i$.}
	\term{\overline{SOC}_i,\underline{SOC}_i}{Available state of charge of storage unit $i$.}
	\term{\overline{P}_{i,t}^{AW},\underline{P}_{i,t}^{AW},P_{i,t}^{AW,av}}{Maximal/minimal/average value of available wind power.}
	\term{P_{i,t}^{AW.h}}{Fluctuation level of available wind generation.}
	\term{\mu_t^{E},\mu_{s,t}^{E}}{Energy market price at time period $t$, as a deterministic coefficient and the value under scenario $s$, respectively.}
	\term{\mu^{RE+},\mu^{RE-}}{Up-/down- reserve energy prices as a deterministic coefficient.}
	\term{\mu_s^{RE+},\mu_s^{RE-}}{Up-/down- reserve energy prices under scenario $s$.}
	\term{\mu_{t}^{RC+},\mu_{t}^{RC-}}{Up-/down- reserve capacity price as a deterministic coefficient.}
	\term{\mu_{s,t}^{RC+},\mu_{s,t}^{RC-}}{Up-/down- reserve capacity price under scenario $s$.}
	\term{\omega_s}{Occurrence probability of scenario $s$.}
\end{termtable}
\subsection*{Variable}
\begin{termtable}[2.5cm]
	\term{u_{i,t},v_{i,t}^{SU},v_{i,t}^{SD}}{Binary variables representing the state/start-up action/shut-down action of CPP $i$.}
	\term{p_{i,t}^{Gen}}{Power generation of CPP $i$.}
	\term{p_{i,t}^D}{Power consumption of flexible demand $i$.}
	\term{p_{i,t}^{ch},p_{i,t}^{dc}}{Charge/discharge power of storage unit $i$.}
	\term{SOC_{i,t}}{State of charge (SoC) of storage unit $i$.}
	\term{p_{i,t}^{W}}{Production of the wind power unit $i$.}
	\term{p_{i,t}^{AW}}{Available wind generation of unit $i$.}
	\term{p_t^{EXCH}}{Energy transaction between the VPP and distribution market at time period $t$.}
	\term{SIG_t^{+},SIG_t^{-}}{Up-/down- regulating signals to the VPP.}
\end{termtable}
\section{Introduction}
%The past few decades have witnessed the proliferation of renewable energy, such as wind generation and photovoltaic (PV), boosted by increasing environmental concerns and energy crisis. However, the intrinsic stochastic, variability and intermittency nature of renewable generation resources imposes great challenges to power system operation\cite{2012Efficient}. 
In recent years, virtual power plant (VPP) technique is developed to promote the effective utilization of renewable resources and achieve environmental and economical superiority\cite{2017AComprehensiveReview}. It combines renewable units with conventional generation units, storage facilities and controllable load demands, etc. Such combination enables distributed energy resources with complementary advantages participating in power system operation and energy-reserve market as an integrated entity. During this process,  uncertainties, distinguished as exogenous and endogenous, are inevitably involved. 
The former, which is also known as decision-independent uncertainty (DIU), is independent of decisions. The latter, which is also known as decision-dependent uncertainty (DDU), can be affected by decision variables. This paper addresses the robust scheduling of a VPP participating day-ahead (DA) energy and reserve market, considering both DIUs and DDUs. 
  
Several closely relevant works are \cite{FuContributing,RiskConstrained,Shabanzadeh,Rahimiyan,PLi,Fourlevel,AStochastic,StochasticAdaptive,DayAhead}, where various optimization techniques are applied to hedge against the risk raised by  uncertainties. In \cite{FuContributing,RiskConstrained}, chance-constrained stochastic programs are utilized to achieve risk-aversion of VPP. In \cite{Shabanzadeh,Rahimiyan,PLi,Fourlevel}, robust optimization (RO) approaches are implemented to maximize the economic profit of VPP under the worst-case realization of the uncertainty in a given set. Reference \cite{Shabanzadeh} applies an RO-based model to the self-scheduling of VPP in the volatile day-ahead market environment whereas the uncertainties pertaining to renewable generations are left out. In \cite{Rahimiyan,Fourlevel}, bidding strategies of VPP in both DA and real-time (RT) markets considering uncertainties of DA market prices, RT market prices and wind production are presented. To hedge against multi-stage uncertainties, a standard two-stage robust model is applied in \cite{Rahimiyan}. Moreover, a four-level robust model is formulated in \cite{Fourlevel} with a tractable algorithm based on strong duality theorem and column-and-constraint generation (C\&CG) algorithm. In \cite{PLi} communication failures and cyber-attacks on the distributed generators in a VPP are considered and a robust economic dispatch of the VPP is accordingly proposed. In \cite{AStochastic,StochasticAdaptive,DayAhead}, the scenario-based stochastic program and the adaptive robust optimization (ARO) are combined, leading to a  stochastic ARO.
%Note that the uncertainties in market clearing prices, renewable generation and reserve deployment requests are all considered as decision independent in these works.
 
In spite of the relevance of the aforementioned literature, the dependency of uncertainties on decisions is disregarded. Specifically, the volatile market prices are regarded as exogenously uncertain as the VPP is assumed to be a price taker in the market. The uncertainties of renewable generations are also considered exogenous since they are determined by uncontrollable natural factors. As for the uncertain reserve deployment requests to VPP, equivalent binary-variable-based representation of the uncertainty set with a given budget parameter indicates that it is a DIU set.
However, when taking into account the reserve energy provided by the VPP, the polyhedral uncertainty set pertaining to reserve deployment requests becomes endogenous, i.e., dependent on VPP's offering in the reserve market, and cannot be reduced to its extreme-based exogenous equivalent. To the best of the authors’ knowledge, no research work has concurrently modeled exogenous uncertainties and endogenous uncertainties for self-scheduling of a VPP in the RO framework, which is specific to this paper.

% chance constrained
%\cite{FuContributing}
%\cite{RiskConstrained}
% robust
%\cite{Shabanzadeh}
%\cite{Rahimiyan}
%\cite{PLi} 
%\cite{Fourlevel}
% stochastic adaptive
%\cite{AStochastic}
%\cite{StochasticAdaptive}
%\cite{DayAhead}

RO under decision-dependent uncertainties (RO-DDU) recently has drawn increasing attention in  the optimization community. Literature regards RO-DDU as two categories: static RO-DDU \cite{Lappas2018,Nohadani2018,Poss2013,vujanic2016,zhang2017} and adaptive RO-DDU (ARO-DDU)\cite{AUnified,Lappas2016,su2020}. In \cite{Lappas2018,Nohadani2018,Poss2013,vujanic2016,zhang2017}, the linear decision-dependency of  polyhedral uncertainty sets on decision variables is considered, rendering a static RO-DDU model. Then, the robust counterpart, which is a mixed integer linear program (MILP), is derived by applying the strong duality theory and  McCormick Envelopes convex relaxation. In \cite{AUnified,Lappas2016,su2020}, ARO-DDU models that concurrently incorporate wait-and-see decisions and endogenous uncertainties are studied. Due to the computational intractability raised by the complex coupling relationship between uncertainties and decisions in two stages, the current works make considerable simplifications on the model. Reference \cite{AUnified,Lappas2016} assume affine decision rules for the wait-and-see decisions, converting the two-stage RO problem into a static RO problem. To address a two-stage ARO-DDU problem without any assumption on affine policies, the extensively-used C\&CG algorithm \cite{BZeng} may fail when the uncertainty set is decision-dependent. In this regard, reference \cite{su2020} focuses on a high-dimensional rectangle DDU set and accordingly proposes an improved C\&CG algorithm with a worst-case-scenario mapping technique.  However, to the best of our knowledge, the solution method for ARO-DDU with general linear dependency has not been addressed.

% static robust ddu
%\cite{Lappas2018}
%\cite{Nohadani2018}
%\cite{Poss2013}
%\cite{vujanic2016}
%\cite{zhang2017}

% adaptive robust ddu
%\cite{AUnified}
%\cite{Lappas2016}

% benders
%\cite{Benders}
% CCG
%\cite{BZeng}

Regarding the aforementioned issues, this paper considers the  robust offering and scheduling strategies of VPP participating in the DA energy-reserve market, where both exogenous and endogenous uncertainties are involved. Specifically, the uncertainties of market prices and renewable generations are exogenous (or called decision-independent), while the uncertainties of reserve deployment requests are endogenous (or called decision-dependent). The main contributions are twofold:
\begin{itemize}
	\item [1)]\textbf{Modeling:} A novel stochastic ARO model incorporating both exogenous and endogenous uncertainties is provided for the robust scheduling of VPP trading in the DA energy-reserve market. Compared with existing works \cite{AStochastic,StochasticAdaptive,DayAhead}, we characterize the dependency of uncertain reserve deployment requests on VPP's decisions in the DA reserve market.
	
	\item [2)]\textbf{Algorithm:} A novel Benders' decomposition based algorithm is proposed to solve the stochastic ARO-DDU problem with general linear decision dependency. The proposed algorithm is guaranteed to converge to the optimum within finite rounds of iterations. To the best of our knowledge, the computational intractability of non-reduced ARO-DDU with general linear decision dependency has not been addressed in the existing literature.
	
\end{itemize}

The rest of this paper is organized as follows. Section II presents the VPP DA robust scheduling formulation with the characterization of both exogenous and endogenous uncertainties. Section III derives the robust counterpart and a solution methodology based on Benders' decomposition. A case study is presented in Section IV. Finally, Section V concludes the paper.

\section{Model Description}
\subsection{DA scheduling of VPP}
Revenue of VPP in the DA energy market comprises the cost of purchasing energy or the income of selling energy and is calculated as follows:
\begin{equation}
R^{NRG}=\sum\nolimits_{t\in{T}}{\mu}_t^{E}p_{t}^{E}
\end{equation}
Reserve market revenue consists of the income of providing reserve service that includes two parts: reserve capacity and reserve energy.
$R^{RSV}=$
\begin{align}
E^{R+}\mu^{RE+}+E^{R-}\mu^{RE-}+\sum\nolimits_{t\in{T}} \left(p^{R+}_t\mu^{RC+}_t
+p^{R-}_t\mu^{RC-}_t\right)
\end{align}

Regarding the generation cost of VPP, operation cost of wind generation units is assumed to be zero, leaving the inherent cost to be the operation cost of CPPs. The operation cost of CPP is computed as
\begin{align}
\label{cost:cpp}
C^{Gen}=\sum\nolimits_{t\in{T},i\in I_G}\left(C_i^{0}u_{i,t}
+C_i^{SU}v_{i,t}^{SU}+C_i^{SD}v_{i,t}^{SD}+C_i^{1}p_{i,t}^{Gen}\right)
\end{align}
which comprises fixed cost, start-up and shut down cost, and the variable generation cost.

The VPP determines the following things as the DA decisions: (i) The power sold to/bought from the day-ahead energy market; (ii) The reserve capacity at each time slots, as well as the maximum reserve energy that can be provided in the day-ahead reserve market; and (iii) The unit commitment of CPP.

\subsection{Uncertainty Characterization}
In this paper, three kinds of uncertainties are taken into consideration as follows.

\subsubsection{Market Clearing Price}
The market clearing prices are exogenously uncertain since the VPP is assumed to be a price taker in DA energy-reserve market. Price uncertainties appear only in the objective function, affecting the optimality of decisions but not the feasibility of the VPP system. Thus it is suitable to model price uncertainty into a scenario-based stochastic programming that aims to minimize the expected net cost of VPP over a set of representative scenarios:  
\begin{subequations}
\begin{align}
\mathbb{E}C^{net} &= \mathbb{E}\left(C^{Gen} - R^{NRG} -R^{RSV}\right)\\
\notag
&= C^{Gen} - \sum\nolimits_{s\in I_S} \omega_s\left(
\mu_s^{RE+}E^{R+}+\mu_s^{RE-}E^{R-}
\right)
\\
&-\sum\nolimits_{s\in I_S,t\in T}\omega_s\left(
\mu_{s,t}^Ep_t^E+\mu^{RC+}_{s,t}p^{R+}_t
+\mu^{RC-}_{s,t}p^{R-}_t
\right)
\end{align}
\end{subequations}

\subsubsection{Available Wind Generation}
Available wind generation $P^{AW}$ is exogenously uncertain since it is determined by nature condition. It appears in the operating constraints of VPP, imposing a significant effect on not only the optimality but also the feasibility of the solution. Thus wind uncertainty is characterized by the following ambiguity set.
\begin{subequations}
	\label{uncertain:W}
\begin{align}
{W}=\left\{{p}_i^{AW}\in\mathbb{R}^T:
\underline{P}_{i,t}^{AW}\le{p}_{i,t}^{AW}\le\overline{P}_{i,t}^{AW},\forall i\in I_W,t\in T\right.\\
\sum\nolimits_{t\in T}\vert{p}_{i,t}^{AW} - P_{i,t}^{AW,av}\vert/P_{i,t}^{AW,h}\le \Gamma_i^{T},\forall i\in I_W\\
\left.\sum\nolimits_{i\in I_W}\vert{p}_{i,t}^{AW} - P_{i,t}^{AW,av}\vert/{P_{i,t}^{AW,h}}\le \Gamma_t^{S},\forall t\in T\right\}
\end{align}
\end{subequations}
where $P_{i,t}^{AW,av}=\frac{\overline{P}_{i,t}^{AW}+\underline{P}_{i,t}^{AW}}{2}$ and $P_{i,t}^{AW,h}=\frac{\overline{P}_{i,t}^{AW}-\underline{P}_{i,t}^{AW}}{2}$, $\forall t\in T,i\in I_W$. It is assumed that the available wind generation fluctuates with the interval between $\underline{P}_{i,t}^{AW}$ and $\overline{P}_{i,t}^{AW}$, under a certain confidence level. $P^{AW,av}$ is the average level for available wind power generation and is calculated as the mean value of the corresponding upper and lower confidence bounds $\underline{P}_{i,t}^{AW}$ and $\overline{P}_{i,t}^{AW}$. $P^{AW,h}$ denotes half of the interval width. To alleviate conservativeness of the model, space robustness budget $\Gamma^{S}$ and time robustness budget $\Gamma^{T}$ is added to avoid that $p^{AW}$ always achieve boundary values.

\subsubsection{Reserve Deployment Request}
Considering the uncertainty in reserve deployment requests $SIG^{+}$ and $SIG^{-}$, energy transaction between the VPP and the distribution energy market $p^{EXCH}$ is endogenously uncertain since it depends upon VPP's decision in DA energy-reserve market. We model the uncertainty of $p^{EXCH}$ by exploring its decision-dependent uncertainty set:
\begin{subequations}
	\label{uncertain:P2}
	\begin{align}
	\label{uncertain:P2:1}
	P(p^E,p^{R+},p^{R-},E^{R+},E^{R-})=\left\{p^{EXCH}\in\mathbb{R}^{\vert T\vert}:\right.\\
	\label{uncertain:P2:2}
	p^{EXCH}_t=p^E_t+ SIG_t^{+}-SIG_t^{-},\forall t\in T\\
	\label{uncertain:P2:3}
	SIG^{+}\in\mathbb{R}^{\vert T\vert},0\le SIG_t^{+}\le p_t^{R+},\forall t\in T\\
	\label{uncertain:P2:4}
	SIG^{-}\in\mathbb{R}^{\vert T\vert},0\le SIG_t^{-}\le p_t^{R-},\forall t\in T\\
	\label{uncertain:P2:5}
	\left.\sum\nolimits_{t\in T} SIG_t^{+} \le E^{R+},\sum\nolimits_{t\in T} SIG_t^{-} \le E^{R-}\right\}
	\end{align}
\end{subequations}
Constraint \eqref{uncertain:P2:5} imposes limits on the total reserve energy that to be deployed. $p^{R+}$, $p^{R-}$, $E^{R+}$, and $E^{R-}$ together control the conservativeness of the ambiguity set associated with the requests for reserve deployment $SIG^{+},SIG^{-}$. Note that in \eqref{uncertain:P2} the complementarity constraint to avoid the situation that up- and down- regulation signals are given simultaneously is omitted. This is because the ambiguity set of $p^{EXCH}$ remains the same with the relaxation on the complementarity constraint.

%\begin{remark}[Decision Dependent Uncertainty Set]
	%The uncertainty set in \eqref{uncertain:P2} is decision-dependent, since the location and shape of the polytope $P$ changes with decisions $p^E,p^{R+},p^{R-},E^{R+},E^{R-}$. Moreover, since $E^{R+}$ and $E^{R-}$ are continuous values, the binary-variable-based representation of uncertainty set widely used in \cite{} is no more applicable.
	
	%the worst-case of $p^{EXCH}$ can not be equivalently characterized by binary variables, imposing difficulty on finding worst-case uncertainty.
%\end{remark}

\subsection{Formulation}
The proposed adaptive robust optimization model aims at minimizing the expected cost over the representative scenarios of market clearing price. Moreover, feasibility of real-time operation of VPP is warranted,  even under the worst-case uncertainties of available wind generation and reserve deployment requests.

\begin{subequations}
\label{problem:RODDU}
\begin{align}
\label{problem:RODDU:obj1}
&\text{minimize}\ \mathbb{E}C^{net},\text{ subject to}\\
\label{constraint1:1}
&\left\{u,v^{SU},v^{SD},p^{E},p^{R+},p^{R-},E^{R+},E^{R-}\right\}\in X\cap X_R\\
&\left\{p^{Gen,0},p^{D,0},p^{ch,0},p^{dc,0},SOC^0,p^{W,0}\right\}
\in Y^0(u,p^{E})
\end{align}
\end{subequations}
where
\begin{subequations}
\label{X}
\begin{align}
\label{constraint1:2}
{ X}:=\left\{u,v^{SU},v^{SD},p^{E},p^{R+},p^{R-},E^{R+},E^{R-}:\right.\\
\label{constraint1:3}
0\le p_t^{R+}\le \overline{P}^{R+},\forall t\in{T}\\
\label{constraint1:4}
0\le p_t^{R-}\le \overline{P}^{R-},\forall t\in{T}\\
\label{constraint1:5}
0\le E^{R+}\le \min\left\{\overline{E}^{R+},\sum\nolimits_{t\in T}p_t^{R+}\right\}\\
\label{constraint1:6}
0\le E^{R-}\le \min\left\{\overline{E}^{R-},\sum\nolimits_{t\in T}p_t^{R-}\right\}\\
\label{constraint1:7}
-EXCH_{\text{max}}\le p_t^E\le EXCH_{\text{max}},\forall t\in{T}\\
% CPP约束
\label{constraint1:8}
u_{i,t},v_{i,t}^{SU},v_{i,t}^{SD}\in\left\{0,1\right\},\forall t\in{T},\forall i\in I_G\\
\label{constraint1:9}
v_{i,t}^{SU}+v_{i,t}^{SD}\le 1,\forall t\in{T},\forall i\in I_G\\
\label{constraint1:10}
u_{i,t+1}=u_{i,t}+v_{i,t}^{SU}-v_{i,t}^{SD},\forall t\in{T},\forall i\in I_G\\
\label{constraint1:11}
-u_{i,t-1}+u_{i,t}\le u_{i,\tau},\forall t\le \tau\le T_i^{\text{on}}+t-1,i\in I_G\\
\label{constraint1:12}
\left.u_{i,t-1}-u_{i,t}+u_{i,\tau}\le 1,\forall t\le \tau\le T_i^{\text{off}}-1,i\in I_G\right\}
\end{align}
\end{subequations}
The feasible region of the wait-and-see decisions is formulated in \eqref{constraint3} where $p^{AW}$ and $p^{EXCH}$ are uncertainties.
\begin{subequations}
	\label{constraint3}
	\begin{align}
	\label{constraint3:2}
	{Y}(u,p^{AW},p^{EXCH}):=\left\{p^{Gen},p^{D},p^{ch},p^{dc},SOC,p^W:\right.\\
	% CCP约束
	\label{constraint3:4}
	u_{i,t}\underline{P}_i^{Gen}\le p_{i,t}^{Gen}\le u_{i,t}\overline{P}_i^{Gen},\forall t\in{T},i\in I_G\\
	\label{constraint3:5}
	p_{i,t+1}^{Gen}-p_{i,t}^{Gen}\le u_{i,t}R_i^{+}+(1-u_{i,t})R_i^{SU},\forall t\in{T},i\in I_G\\
	\label{constraint3:6}
	p_{i,t-1}^{Gen}-p_{i,t}^{Gen}\le u_{i,t}R_i^{-}+(1-u_{i,t})R_i^{SD},\forall t\in{T},i\in I_G\\
	% 负荷约束
	\label{constraint3:7}
	\underline{P}^D_{i,t}\le p_{i,t}^D\le \overline{P}^D_{i,t},\forall t\in{T},\forall i\in I_D\\
	\label{constraint3:75}
	-r_i^{D-}\le p_{i,t+1}^D-p_{i,t}^D\le r_i^{D+},\forall t\in T,\forall i\in I_D\\
	\label{constraint3:8}
	\sum\nolimits_{t\in{T}}p_{i,t}^D\ge \underline{D}_i^D,\forall i\in I_D\\
	% 储能约束
	\label{constraint3:9}
	0\le p^{ch}_{i,t}\le \overline{P}_{i}^{ch},\forall t\in{T},\forall i\in I_{ES}\\
	\label{constraint3:10}
	0\le p^{dc}_{i,t}\le \overline{P}_{i}^{dc},\forall t\in{T},\forall i\in I_{ES}\\
	\label{constraint3:11}
	SOC_{i,t}=SOC_{i,t-1}+\eta_i^{ch}p_{i,t}^{ch}-\frac{1}{\eta_i^{dc}}p_{i,t}^{dc},\forall t\in{T},i\in I_{ES}\\
	\label{constraint3:12}
	\underline{SOC}_i\le SOC_{i,t}\le \overline{SOC}_i,\forall i\in I_{ES}\\
	% 风电约束
	\label{constraint3:13}
	0\le p_{i,t}^{W}\le {p}_{i,t}^{AW},\forall t\in{T},\forall i\in I_{W}\\
	\notag
	\sum\nolimits_{i\in I_G}p_{i,t}^{Gen}+\sum\nolimits_{i\in I_W}p_{it}^{W}+\sum\nolimits_{i\in I_{ES}}p_{i,t}^{dc}=p_t^{EXCH}+\\
	\label{constraint3:15}
	\left.\sum\nolimits_{i\in I_D}p_{i,t}^D+\sum\nolimits_{i\in I_{ES}}p_{i,t}^{ch},\forall t\in {T}
	\right\}.
	\end{align} 
\end{subequations}
Thus the feasible region of the baseline re-dispatch decisions $p^{Gen,0},p^{D,0},p^{ch,0},p^{dc,0},SOC^0,p^{W,0}$ is 
\begin{align}
\notag
&Y^0(u,p^{E}):=Y(u,p^{AW,av},p^{E})
\end{align}
Then the robust feasibility set of $x$ is defined as
\begin{subequations}
\label{X_R}
\begin{align}
{X}_R := \left\{u,p^{E},p^{R+},p^{R-},E^{R+},E^{R-}:\right.\\
\forall p^{AW}\in {W} \text{ defined in \eqref{uncertain:W}},\\
\forall p^{EXCH}\in{P}({ p^E,p^{R+},p^{R-},E^{R+},E^{R-}})\text{ defined in \eqref{uncertain:P2}},\\
\left.{Y}(u,p^{AW},P^{EXCH})\neq\emptyset\right\}
\end{align}
\end{subequations}

The feasible region of VPP's DA decisions ($u$, $v^{SU}$, $v^{SD}$, $p^{E}$, $p^{R+}$, $p^{R-}$, $E^{R+}$, $E^{R-}$) is denoted by ${X}$, with the specific form of \eqref{constraint1:2}-\eqref{constraint1:12}, including
constraints of the DA market \eqref{constraint1:2}-\eqref{constraint1:7} which impose limits on the energy and reserve offering of the VPP, as well as constraints of CPP \eqref{constraint1:8}-\eqref{constraint1:12}. The feasible region of VPP's RT decisions ($p^{Gen}$, $p^{D}$, $p^{ch}$, $p^{dc}$, $SOC$, $p^W$), also called wait-and-see decisions, is denoted by $Y$, where constraints of CPP \eqref{constraint3:4}-\eqref{constraint3:6}, flexible demand \eqref{constraint3:7}-\eqref{constraint3:8}, the energy storage unit \eqref{constraint3:9}-\eqref{constraint3:12}, the wind generation unit \eqref{constraint3:13}, and the power balance of VPP \eqref{constraint3:15} are included. When no uncertainties exist, i.e., $p^{AW}=p^{AW,av},p^{EXCH}=p^E$, the feasible region of the baseline re-dispatch decisions $p^{Gen,0},p^{D,0},p^{ch,0},p^{dc,0},SOC^0,p^{W,0}$ is denoted by $Y^0$. The uncertain parameters in the DA scheduling problem are $p^{AW}$ and $p^{EXCH}$. The decision independent uncertainty set ${W}$ for $p^{AW}$ is given in \eqref{uncertain:W} and the decision dependent uncertainty set ${P}(\cdot)$ for $p^{EXCH}$ is given in \eqref{uncertain:P2}. For the wait and see decisions $p^{Gen},p^{D},p^{ch},p^{dc},SOC,p^W$, their feasible space ${Y}(u,p^{AW},p^{EXCH})$ is actually a set-valued map parameterized by the first stage decision ${u}$ and the uncertain variables $p^{AW},p^{EXCH}$. Besides the aforementioned operation constraints of each stage, the first-stage decision $u$, $v^{SU}$, $v^{SD}$, $p^{E}$, $p^{R+}$, $p^{R-}$, $E^{R+}$, $E^{R-}$ has to satisfies robust feasibility, as characterized in \eqref{X_R}. $X_R$ is called robust feasibility region.

Next we give the compact form of two-stage stochastic robust optimization problem \eqref{problem:RODDU}-\eqref{X_R}.
To simplify the formulation, the following terminology is adopted. 
\begin{subequations}
	\begin{align}
	&x:=\left\{u,v^{SU},v^{SD},p^{E},p^{R+},p^{R-},E^{R+},E^{R-}\right\}\\
	&y^0:=\left\{p^{Gen,0},p^{D,0},p^{ch,0},p^{dc,0},SOC^0,p^{W,0}\right\}\\
	&w:=\left\{p^{AW},p^{EXCH}\right\}\\
	&{y}:=\left\{p^{Gen},p^{D},p^{ch},p^{dc},SOC,p^W\right\}
	\end{align}
\end{subequations}
where 
${x}\in\mathbb{R}^{n_{\mathcal{R}}}\times\mathbb{Z}^{n_{\mathcal{Z}}}$, $w\in \mathbb{R}^{n_w}$ and $y,y^0\in \mathbb{R}^{n_y}$. The dimension of ${x}$ is denoted by $n_x=n_{\mathcal{R}}+n_{\mathcal{Z}}$. We denote the cost items in \eqref{problem:RODDU:obj1} by a unified form $f({x},y^0):\mathbb{R}^{n_{\mathcal{R}}+n_y}\times\mathbb{Z}^{n_{\mathcal{Z}}}\rightarrow \mathbb{R}^1$. Then  \eqref{problem:RODDU}-\eqref{X_R} is formulated in a compact form of
\begin{subequations}
	\label{compact}
	\begin{align}
	\label{compact:1}
	&\min\nolimits_{{x},y^0}f({x},y^0)\\
	\label{compact:2}
	&\text{s.t.}\ {x}\in {X}\cap X_R,y^0\in Y^0(x)\\
	\label{compact:3}
	&{X}_R = \left\{{x}|\forall {w}\in \mathcal{W}({ x}),\mathcal{Y}({x},{w})\neq \emptyset\right\}\\
	\label{compact:4}
	&\mathcal{W}({x}) = \left\{{w}\in\mathbb{R}^{n_w}|{G}{ w}\le {g}+{\Delta} {x}\right\}\\
	\label{compact:5}
	&\mathcal{Y}({x},{w})=\left\{{y}\in\mathbb{R}^{n_y}|{A}{x}+B{y}+C{w}\le {b},{y}\ge {0}\right\}
	\end{align}
\end{subequations}
where ${G}\in\mathbb{R}^{r\times n_w},{g}\in\mathbb{R}^{r},{\Delta} \in \mathbb{R}^{r\times n_x},A\in\mathbb{R}^{m\times n_x},B\in\mathbb{R}^{m\times n_y},C\in\mathbb{R}^{m\times n_w}$ and ${ b}\in\mathbb{R}^{m}$ are constants. $\mathcal{W}(x)$ is a unified form of the decision-independent uncertainty set ${W}$ in \eqref{uncertain:W} and the decision dependent uncertainty set ${P}$ in \eqref{uncertain:P2}. 
% 想说明这是一个general model
Note that \eqref{compact:4} models general decision dependence, which encompasses the case of decision-independent uncertainties by setting the corresponding rows of $\Delta$ to zeros. $\mathcal{Y}({x},{w})$ is the compact form of $Y$ in \eqref{constraint3}.

Problem \eqref{compact} is a two-stage adaptive robust optimization problem with decision dependent uncertainties. Regarding the solution methodology to this type of problem, the C\&CG algorithm is no longer applicable, for the reason that the worst-case uncertainty ${w}^*\in \mathcal{W}({x}^1)$ with a given ${x}^1$ may lie outside the uncertainty set when giving another ${x}^2$, i.e., ${w}^*\notin \mathcal{W}({x}^2)$. Then the feasibility cut of the C\&CG algorithm may fail to obtain an optimal solution. Moreover, since the vertices set of polytope $\mathcal{W}({x})$ changes with ${x}$, the C\&CG algorithm no longer guarantees finite iterations to convergence.

\section{Solution Methodology}
%We solve the equivalent problem of TRO-DDU problem \eqref{compact}, which is given in \eqref{equivalent}, using an improved benders decomposition algorithm, which involves the iterative solution of a master problem and a robust feasibility examination problem. The algorithm is proved to converge to the optimum of \eqref{equivalent} within finite iterations.

\subsection{Equivalent Transformation}
Given a first stage decision ${x}$, the robust feasibility of ${x}$, i.e., whether ${x}$ locates within ${X}_R$, can be examined by solving the following relaxed bi-level problem:
\begin{subequations}
	\label{def:1:6}
	\begin{align}
	\label{def:1:6:1}
	&R({x})=\max\nolimits_{{w}\in\mathcal{W}({x})}\min\nolimits_{{y},{s}}{ 1}^{\mathsf{T}}{ s}\\
	\label{def:1:6:2}
	&\text{s.t.}\ A{x}+B{y}+C{w}\le {b}+{s},{y}\ge {0}, {s}\ge {0}
	\end{align}
\end{subequations}
where ${s}\in\mathbb{R}^m$ is the supplementary variable introduced to relax the constraint $A{ x}+B{ y}+C{ w}\le { b}$ in $\mathcal{Y}({ x},{ w})$. If $R({x})\le 0$, ${ x}$ is robust feasible, i.e., ${ x}\in X_R$. Else if $R({ x})> 0$, there exists a realization of the uncertain ${ w}$ lying in the $\mathcal{W}({ x})$ that makes no feasible second-stage decision ${ y}$ is available. Since ${ x}\in X_R$ if and only if $R({ x})\le 0$, we substitute the constraint $x\in X_R$ in \eqref{compact} by $R({ x})\le 0$.

It is useful to write the dual of the inner minimization problem in $R({x})$. Then, $R({x})$ can be equivalently transformed into the following single-level bi-linear maximization problem
\begin{subequations}
	\label{ineq:alg:2}
	\begin{align}
	R({x})=&\max\nolimits_{{w},{ \pi}} {\pi}^{\mathsf{T}}\left({ b}-A{ x}-C{ w}\right)\\
	&{\rm{s.t.}}\ {\pi}\in \Pi,\ { w}\in \mathcal{W}({ x})
	\end{align}
\end{subequations}
where ${\pi}\in\mathbb{R}^m$ is the dual variable on constraint \eqref{def:1:6:2} and $\Pi=\left\{{ \pi}|B^{\mathsf{T}}{ \pi}\le { 0},-{ 1}\le{\pi}\le { 0}\right\}$. Therefore, problem \eqref{compact} can be reformulated into the following non-linear static robust optimization problem with DDU:
\begin{subequations}
	\label{equivalent}
	\begin{align}
	\label{equivalent:1}
	&\min\nolimits_{{x},y^0} f({x},y^0) \\
    \label{equivalent:2}
	&{\rm{s.t.}}\ {x}\in { X},y^0\in Y^0(x)\\
    \label{equivalent:3}
	&{0}\ge{ \pi}^{\mathsf{T}}({ b}-A{ x}-C{w}),\forall \pi\in \Pi,w\in \mathcal{W}({x})
	\end{align}
\end{subequations}
Constraint \eqref{equivalent:3} is decision-dependent static robust constraint. However, due to the bi-linear relationship between variable $\pi$ and variable $w$ in term $-{\pi}^{\mathsf{T}}C{w}$, techniques used to derive a robust counterpart of regular static robust optimization are no more applicable to problem  \eqref{equivalent}. To address the difficulty in solving ARO-DDU problem \eqref{compact} and its equivalent formulation \eqref{equivalent}, next we provide a novel two-level iterative solution algorithm based on Benders decomposition\cite{Benders}.
\subsection{Master Problem (MP)}
The master problem at iteration $k$ is formulated below:
\begin{subequations}
	\label{ineq:alg:5}
	\begin{align}
	\label{ineq:alg:5:1}
	&\min\nolimits_{{x},y^0} f({x},y^0)\\
	\label{ineq:alg:5:2}
	&{\rm{s.t.}}\ {x}\in {X},y^0\in Y^0(x)\\
	\label{fea_cut}
	&0 \ge {{\pi}_j^*}^{\mathsf{T}}\left({ b}-A{ x}-C{ w}\right),\forall { w}\in\mathcal{W}({ x}),j\in [k]
	\end{align}
\end{subequations}
where ${\pi}_1^*,...,{\pi}_{k}^*$ are solutions from the robust feasibility examination problem. If ${\pi}_1^*,...,{\pi}_{k}^*\in \Pi$, then the MP \eqref{ineq:alg:5} is a relaxation to \eqref{equivalent}. We solve MP \eqref{ineq:alg:5} to derive a relaxed optimum of \eqref{equivalent}. Constraints \eqref{fea_cut} are feasibility cuts to MP. They are designed to have the following salient features: (i) The worst-case uncertainty ${w}^*$ is not involved, to accommodate the coupling relation between ${x}$ and ${w}$, which is different from the C\&CG algorithm. (ii) Dual information of robust feasibility examination problem (i.e., ${ \pi}^*$) are included, inspired by the Benders dual decomposition. However, they are designed to be no longer a hyperplane, but a static robust constraint, to comprise a cluster of worst-case uncertainties.

Next, we illustrate how to deal with the robust constraint \eqref{fea_cut} by substituting it with its robust counterpart. For any given $j$ in $[k]$, constraint  \eqref{fea_cut} is equivalent to
	\begin{align}
	\label{ineq:impl:1}
	0 \ge {{\pi}_j^*}^{\mathsf{T}}\left({ b}-A{ x}\right)+
	\left\{
	\begin{array}{l}
	\max_{{w}_j}-{{u}_j^*}^{\mathsf{T}}C{ w}_j\\
	\text{s.t.}\ G{ w}_j\le { g} + \Delta { x}
	\end{array}\right\}
	\end{align}
We deploy the KKT conditions of the inner-level problem in \eqref{ineq:impl:1} as follows
\begin{subequations}
	\label{ineq:impl:4}
	\begin{align}
	\label{ineq:impl:4:1}
	& G^{\mathsf{T}}{ \lambda}_j=-C^{\mathsf{T}}{ \pi}_j^*\\
	\label{ineq:impl:4:2}
	& { \lambda}_j\ge { 0}\perp G{ w}_j\le { g} + \Delta { x}
	\end{align}
\end{subequations}
where ${\lambda}_j\in \mathbb{R}^r$ is the corresponding dual variable  and \eqref{ineq:impl:4:2} denotes the complementary relaxation conditions. The non-linear complementary conditions \eqref{ineq:impl:4:2} can be exactly linearized through big-M method by introducing the binary supplementary variable ${ z}_j\in\left\{0,1\right\}^{r}$ and a sufficiently large positive number $M$ as follows:
\begin{subequations}
	\label{linear}
	\begin{align}
		\label{linear:1}
	&{ 0}\le { \lambda}_j\le M(1-{ z}_j)\\
		\label{linear:2}
	&{ 0}\le { g} + \Delta { x}-G{ w}_j\le M{ z}_j
	\end{align}
\end{subequations}
Then the MP \eqref{ineq:alg:5} has the following robust counterpart which is a MILP problem.
\begin{subequations}
	\label{ineq:impl:6}
	\begin{align}
	&\min\nolimits_{{x},y^0,{z},{\lambda},{w}} f({x},y^0)\\
	&{\rm{s.t.}}\ {x}\in X,y^0\in Y^0(x)\\
	&\left.
	\begin{array}{ll}
	&0 \ge {{\pi}_j^*}^{\mathsf{T}}\left({ b}-A{ x}\right)- {{\pi}_j^*}^{\mathsf{T}}C{ w}_j\\
    &\text{\eqref{ineq:impl:4:1}, \eqref{linear:1}, \eqref{linear:2}}\\
	&{z}_j\in\left\{0,1\right\}^r,{ \lambda}_j\in\mathbb{R}^r,{ w}_j\in\mathbb{R}^{n_w}
	\end{array}
	\right\}j\in [k]
	\end{align}
\end{subequations}

\subsection{Robust Feasibility Examination Subproblem}
The subproblem in this subsection examines the robust feasibility of given $x^k$ by solving $R(x^k)$. $R(x)$ and its equivalent form are given in \eqref{def:1:6} and \eqref{ineq:alg:2}, respectively. The bi-linear objective item $-{\pi}^{\mathsf{T}}C{w}$ imposes difficulties on solving $R({x})$. Next we provide linear surrogate formulations of $R({x})$. 

The robust feasibility examination problem $R({ x})$ in \eqref{ineq:alg:2} can be equivalently written into
	\begin{align}
	\label{ineq:impl:8}
	R({x})=\max\nolimits_{{\pi}\in \Pi} \left\{{ \pi}^{\mathsf{T}}\left({ b}-A{ x}\right)+
	\begin{array}{l}
	\max_{{ w}}-{\pi}^{\mathsf{T}}C{ w}\\
	{\rm{s.t.}}\ G{ w}\le { g}+\Delta{ x}
	\end{array}
	\right\}.
	\end{align}
Then we deploy the KKT conditions of the inner-level problem, which are
\begin{subequations}
	\label{ineq:impl:9}
	\begin{align}
	\label{ineq:impl:9:1}
	&-{\pi}^{\mathsf{T}}C{ w}= ({ g}+\Delta{ x})^{\mathsf{T}}{ \zeta}\\
	\label{ineq:impl:9:2}
	&{ \zeta}\ge{ 0}\perp G{ w}\le { g}+\Delta{ x}\\
	\label{ineq:impl:9:3}
	&G^{\mathsf{T}}{ \zeta}=-C^{\mathsf{T}}{\pi}
	\end{align}
\end{subequations}
where ${ \zeta}\in\mathbb{R}^{r}$ is the corresponding dual variable. The complementary constraint \eqref{ineq:impl:9:2} can be linearlized by introducing binary supplementary variable ${v}\in\left\{0,1\right\}^{r}$ like what we do to \eqref{ineq:impl:4:2}. Moreover, since strong duality holds, we substitute $-{\pi}^{\mathsf{T}}C{ w}$ by $({g}+\Delta{x})^{\mathsf{T}}{ \zeta}$. 
Then, the subproblem $R(x)$ can be equivalently transformed into the following MILP
\begin{subequations}
\label{ineq:impl:10}
	\begin{align}
	R(x)=&\max\nolimits_{{\pi},{ w},{\zeta},v} { \pi}^{\mathsf{T}}({ b}-A{ x})+({ g}+\Delta { x})^{\mathsf{T}}{\zeta}\\
	&\text{s.t.}\ \pi\in \Pi,\text{ \eqref{ineq:impl:9:3}},\\
	&{ 0}\le{ \zeta}\le M(1-{v})\\
	&{ 0}\le { g}+\Delta{x}-G{w}\le M{v}\\
	&{ v}\in\left\{0,1\right\}^{r},{ \zeta}\in\mathbb{R}^{r}
	\end{align}
\end{subequations}

\subsection{Modified Benders Decomposition Algorithm}
Now we have the overall iterative algorithm, as given in Algorithm \ref{alg}. Convergence and optimality of the Algorithm \ref{alg} are justified by Theorem \ref{theorem:0}. Theorem \ref{theorem:0} indicates that the proposed modified Benders decomposition method can find the optimal solution of ARO-DDU problem \eqref{compact} within finite steps. Proof of Theorem \ref{theorem:0} is given in the Appendix.
\begin{algorithm}[htb]
	\caption{Modified Benders Decomposition Algorithm}
	\label{alg}
	% step 0
	{\bf Step 0: Initialization}\\
	Set $k=0$. Choose an initial solution ${x}^k\in X,y^{0,k}\in Y^0(x^k)$.\\
	% step 1
	{\bf Step 1: Robust Feasibility Examination}\\
	Check robust feasibility of ${x}^k$ by solving $R({ x}^k)$ in \eqref{ineq:impl:10}.
	 Let $({w}_k^*,{\pi}_k^*)$ be the optimum of $R({ x}^k)$. If $R({ x}^k)>0$, $k=k+1$, then go to Step 2. Else if $R({ x}^k)=0$, terminate the algorithm and output the optimal solution $(x^k,y^{0,k})$.\\
	% step 2
    {\bf Step 2: Solve Master Problem (MP)}\\
    Solve the master problem \eqref{ineq:impl:6}. Let $({x}^{k},y^{0,k})$ be the optimum and then go to Step 1.
\end{algorithm}
\begin{theorem}
	\label{theorem:0}
	Let $p$ be the number of extreme points of $\Pi$. Then the Algorithm \ref{alg}  generates an optimal solution to \eqref{compact} in $\mathcal{O}(p)$ iterations.
\end{theorem}

\section{Case Studies}
In this section, case studies are conducted on MATLAB with a laptop with Intel i5-8250U 1.60GHz CPU and 4GB of RAM. GUROBI 9.1.0 is used as the solver.

\subsection{Setup}
We consider a VPP that consists of four conventional generators, a wind farm, an energy storage facility, a flexible load, and three fixed loads. The schematic diagram of the VPP is given in Fig.\ref{fig1}. For the DA robust scheduling of the VPP, 24 hourly periods are considered, i.e., $\vert T\vert=24$.
\begin{figure}[htb]
	\centering
	\includegraphics[width=0.8\linewidth]{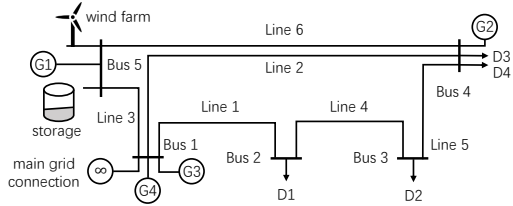}
	\caption{5-bus network.}
	\label{fig1}
\end{figure}

Parameters of the four conventional generators are provided in Table.\ref{table:gen}. The 400MW wind farm is located at Bus 5, and the confidence bounds and average levels for available wind power generation are illustrated in Fig.\ref{fig2_3}. Technical data of the fixed and flexible loads are provided in Table.\ref{table:load}, and the daily profiles of the total fixed load are shown in Fig.\ref{fig2_3}. The storage facility is located at Bus 5, with a  capacity of 100 MW/200 MW.h and conversion efficiency of 90\%. The maximum and minimum SoC are 180MW.h and 20MW.h, respectively.

The VPP is connected to the main grid at Bus 1. The maximum power that can be obtained from or sold to the main grid is 400 MW. The maximum participation in DA reserve market is 250 MW at each time slot, for both up- and down- reserve market. The maximum deployed reserve energy is 6000 MW.h ($250\ {\rm MW}\times 24\ {\rm h}$), for both up- and down- reserve deployment requests. Market price scenarios are generated from Nord Pool price data from October 25th to November 25th, 2020\cite{NordPool}, through K-means clustering. Therefore, the uncertain market prices are represented by 8 typical equiprobable scenarios.
\begin{table}[htb]
	\caption{Parameters of conventional generators.}
	\label{table:gen}
	\setlength{\tabcolsep}{3pt}
	\centering
	\begin{tabular}{m{15pt}<{\centering}m{20pt}<{\centering}m{35pt}<{\centering}m{25pt}<{\centering}m{25pt}<{\centering}m{25pt}<{\centering}m{25pt}<{\centering}m{25pt}<{\centering}}
		\toprule
        &
        Loc &
        \begin{tabular}[c]{@{}c@{}}{[}$\underline{P}_i^{Gen},\overline{P}_i^{Gen}${]}\\ (MW)\end{tabular} &
        \begin{tabular}[c]{@{}c@{}}$R_i^+,R_i^-$,\\ $R_i^{SU},R_i^{SD}$\\ (MW)\end{tabular} &
        \begin{tabular}[c]{@{}c@{}}$T_i^{on}$,\\ $T_i^{off}$\\ (hour)\end{tabular} &
        \begin{tabular}[c]{@{}c@{}}$C_i^{SU}$,\\ $C_i^{SD}$\\ (\$/times)\end{tabular} &
        \begin{tabular}[c]{@{}c@{}}$C_i^{Gen,0}$\\ (\$/h)\end{tabular} &
        \begin{tabular}[c]{@{}c@{}}$C_i^{Gen,1}$\\ (\$/MW.h)\end{tabular} \\[3pt]
		\hline
		G1 & Bus 5 & {[}200,400{]} & 50  & 6 & 100 & 50 & 40 \\[3pt]
		G2 & Bus 4 & {[}150,300{]} & 50  & 5 & 100 & 50 & 60 \\[3pt]
		G3 & Bus 1 & {[}150,250{]} & 100 & 4 & 100 & 50 & 70 \\[3pt]
		G4 & Bus 1 & {[}250,500{]} & 80  & 6 & 100 & 50 & 50\\[3pt]
		\bottomrule
	\end{tabular}
\end{table}
\begin{table}[!ht]
	\caption{Parameters of Load.}
	\label{table:load}
	\setlength{\tabcolsep}{3pt}
	\centering
	\begin{tabular}{ccccccc}
		\toprule
		&
		Loc &
		Type &
		Ratio &
		\begin{tabular}[c]{@{}c@{}}$\underline{D}_i^D$\\ (MWh)\end{tabular} &
		\begin{tabular}[c]{@{}c@{}}{[}$\underline{P}_i^D,\overline{P}_i^D${]}\\ (MW)\end{tabular} &
		\begin{tabular}[c]{@{}c@{}}$r_i^{D-},r_i^{D+}$\\ (MW)\end{tabular} \\[3pt]
		\hline
		D1 & Bus 2 & fixed    & 0.3 & -    & -           & -   \\[3pt]
		D2 & Bus 3 & fixed    & 0.3 & -    & -           & -   \\[3pt]
		D3 & Bus 4 & fixed    & 0.4 & -    & -           & -   \\[3pt]
		D4 & Bus 4 & flexible & -   & 1500 & {[}0,200{]} & 110\\[3pt]
		\bottomrule
	\end{tabular}
\end{table}
\begin{figure}[!ht]
	\centering
	\includegraphics[width=0.9\linewidth]{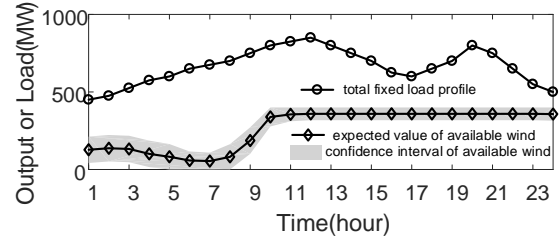}
	\caption{Total fixed load profile; Confidence bounds and average levels for available wind power generation.}
	\label{fig2_3}
\end{figure}

\subsection{Baseline Results}
\label{results}
In this subsection, wind uncertainty budgets are fixed as $\Gamma^T = 8,\Gamma^S=1$. We solve the stochastic robust scheduling problem of VPP by the proposed Algorithm.\ref{alg}. The algorithm converges after 25 iteration rounds, the evolution process of which is depicted in Fig.\ref{fig6}-\ref{fig7}. The increasing net cost and the diminishing reserve revenue represents VPP's hedging against the worst-case realization of uncertainties concerning available wind generation and reserve deployment requests.

\begin{figure}[htb]
	\centering
	\includegraphics[width=0.9\linewidth]{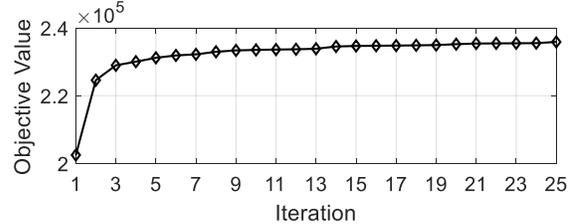}
	\caption{Evolution of objective value with the number of iterations.}
	\label{fig6}
\end{figure}

\begin{figure}[htb]
	\centering
	\includegraphics[width=0.9\linewidth]{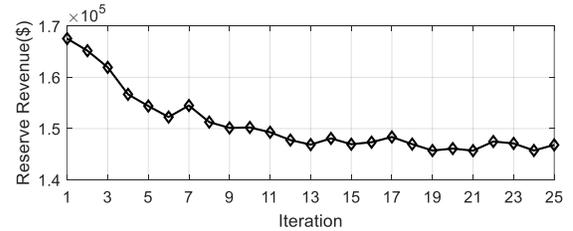}
	\caption{Evolution of reserve revenue with the number of iterations.}
	\label{fig8}
\end{figure}

\begin{figure}[htb]
	\centering
	\includegraphics[width=0.9\linewidth]{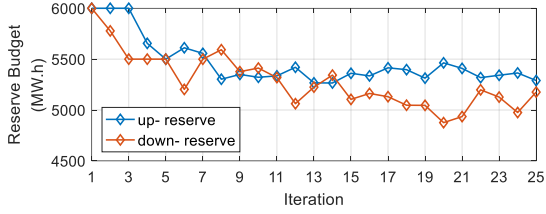}
	\caption{Evolution of reserve budget with the number of iterations.}
	\label{fig7}
\end{figure}

\subsection{Sensitivity Analysis}
\subsubsection{Impact of Wind Uncertainty Budget}
In this case, 7 time budgets of wind uncertainty $\Gamma^T$ from 0 to 12 with a gradient of 2 are introduced. Since there is only one wind farm in the VPP, impact of space robustness parameter $\Gamma^S$ is omitted since $\Gamma^S$ is fixed as 1. Participation in DA energy-reserve market in terms of different $\Gamma^T$ is provided in Table.\ref{table:GammaT}. It is observed that the VPP robust scheduling decisions exhibit different tendencies in DA energy and reserve market when responding to different wind uncertainty budget: as the value of $\Gamma^T$ increases, the amount of reserve offering provided in DA market trends to decrease, while there is no obvious trend for the energy offering in DA market. This is because the VPP would like to keep more ramping resources inside the VPP to hedge against the increasing uncertainty of available wind generation power. An increasing net cost of VPP can also be observed as the value of $\Gamma^T$ increases, indicating that a larger uncertainty set always comes with a higher price of robustness.

\begin{table}[htb]
	\caption{Impact of robustness parameter $\Gamma^T$ on VPP's cost and revenue.}
	\label{table:GammaT}
	\setlength{\tabcolsep}{3pt}
	\centering
	\begin{tabular}{ccccc}
		\toprule
		$\Gamma^T$ &
		\begin{tabular}[c]{@{}c@{}}Net \\ Cost(\$)\end{tabular} &
		\begin{tabular}[c]{@{}c@{}}Energy Market \\ Revenue(\$)\end{tabular} &
		\begin{tabular}[c]{@{}c@{}}Reserve Market \\ Revenue(\$)\end{tabular} &
		\begin{tabular}[c]{@{}c@{}}CPP \\ Cost(\$)\end{tabular} \\[3pt]
		\hline
		0  & 226438.18 & 273204.95 & 154713.67 & 654356.80 \\[3pt]
		2  & 230191.26 & 291456.07 & 152588.30 & 674235.63 \\[3pt]
		4  & 231978.20 & 270391.64 & 149658.90 & 652028.74 \\[3pt]
		6  & 233943.76 & 272301.82 & 148452.17 & 654697.74 \\[3pt]
		8  & 235871.82 & 274171.84 & 146816.15 & 656859.81 \\[3pt]
		10 & 236637.36 & 268425.54 & 144729.88 & 649792.77 \\[3pt]
		12 & 239033.89 & 269952.13 & 144610.50 & 653596.51\\[3pt]
		\bottomrule
	\end{tabular}
\end{table}

\subsubsection{Impact of $\mu^{RE+}$ and $\mu^{RE-}$}
In this case, we present the impact of upward and downward reserve energy price $\mu^{RE+}$ and $\mu^{RE-}$ on the reserve offering behavior of VPP. The results are displayed in Table.\ref{Table:muRE+}-\ref{Table:muRE-}. It is observed that as $\mu^{RE+}$($\mu^{RE-}$) increases, the amount of reserve capacity and reserve energy trend to increase. Certainly, the reserve deployment uncertainty would rise accordingly,  but since the reserve revenue is high, the VPP would like to sacrifice more in DA energy market or pay more for CPP generation cost to hedge against a severer realization of the worst-case reserve deployment. Conversely, if the value of $\mu^{RE+}$ and $\mu^{RE-}$ are relatively small, the VPP trends to slash the reserve offering directly to restrict the uncertainty and ensure robust feasibility.

\begin{table}[htb]
	\caption{Impact of $\mu^{RE+}$ on VPP's cost and revenue with fixed $\mu^{RE-}=8$.}
	\label{Table:muRE+}
	\setlength{\tabcolsep}{3pt}
	\centering
	\begin{tabular}{cccccc}
		\toprule
		$\mu^{RE+}$ &
		\begin{tabular}[c]{@{}c@{}}Net\\ Cost(\$)\end{tabular} &
		\begin{tabular}[c]{@{}c@{}}Energy \\ Market\\ Revenue(\$)\end{tabular} &
		\begin{tabular}[c]{@{}c@{}}Reserve \\ Capacity\\ Revenue(\$)\end{tabular} &
		\begin{tabular}[c]{@{}c@{}}Reserve \\ Energy\\ Revenue(\$)\end{tabular} &
		\begin{tabular}[c]{@{}c@{}}CPP\\ Cost(\$)\end{tabular} \\[3pt]
		\hline
		6  & 246468.82 & 267631.40 & 62422.43 & 72524.55  & 649047.20 \\[3pt]
		8  & 235871.82 & 274171.84 & 63090.85 & 83725.30  & 656859.81 \\[3pt]
		10 & 225351.63 & 266644.06 & 63187.46 & 93450.33  & 648633.48 \\[3pt]
		12 & 213793.07 & 253643.08 & 65379.32 & 108422.57 & 641238.04 \\[3pt]
		14 & 202718.58 & 247247.32 & 65694.42 & 119613.80 & 635274.11\\[3pt]
		\bottomrule
	\end{tabular}
\end{table}

\begin{table}[htb]
	\caption{Impact of $\mu^{RE-}$ on VPP's cost and revenue with fixed $\mu^{RE+}=8$.}
	\label{Table:muRE-}
	\setlength{\tabcolsep}{3pt}
	\centering
	\begin{tabular}{cccccc}
		\toprule
		$\mu^{RE-}$ &
		\begin{tabular}[c]{@{}c@{}}Net\\ Cost(\$)\end{tabular} &
		\begin{tabular}[c]{@{}c@{}}Energy \\ Market\\ Revenue(\$)\end{tabular} &
		\begin{tabular}[c]{@{}c@{}}Reserve \\ Capacity\\ Revenue(\$)\end{tabular} &
		\begin{tabular}[c]{@{}c@{}}Reserve \\ Energy\\ Revenue(\$)\end{tabular} &
		\begin{tabular}[c]{@{}c@{}}CPP\\ Cost(\$)\end{tabular} \\[3pt]
		\hline
		6  & 246081.71 & 263685.24 & 62631.36 & 72210.65  & 644608.96 \\[3pt]
		8  & 235871.82 & 274171.84 & 63090.85 & 83725.30  & 656859.81 \\[3pt]
		10 & 225589.90 & 276759.48 & 63254.95 & 95242.94  & 660847.28 \\[3pt]
		12 & 213779.46 & 281404.67 & 62394.43 & 107333.27 & 664911.84 \\[3pt]
		14 & 203425.38 & 289312.44 & 63227.29 & 119525.46 & 675490.57\\[3pt]
		\bottomrule
	\end{tabular}
\end{table}

\subsection{Comparative Performance Study}
\subsubsection{Comparison between DIU and DDU}
To present the competitiveness of the proposed decision-dependent uncertain regulating signal formulation \eqref{uncertain:P2}, a decision-independent formulation is introduced in \eqref{DIU} as a reference case. In \eqref{DIU}, $V$ is a decision-independent set where $v_t^{R+}$ and $v_t^{R-}$ are the binary variables to model the worst-case upward and downward reserve deployment request, respectively. $\Gamma^R\in\left\{0,1,...,24\right\}$ is the reserve uncertainty budget parameter which controls the conservativeness of the model in \eqref{DIU} and is pre-determined before the robust scheduling of VPP.

\begin{subequations}
\label{DIU}
\begin{align}
p^{EXCH}_t=p^E_t+ v_t^{R+}p_t^{R+} + v_t^{R-}p_t^{R-},\forall t\in T
\end{align}
where
\begin{align}
v_t^{R+},v_t^{R-}\in\label{V}
V:=\left\{
\begin{array}{l}
v^{R+},v^{R-}\in\left\{0,1\right\}^{\vert T\vert}:\\
v_t^{R+} + v_t^{R-}\le 1,\forall t\in T\\
\sum_{t=1}^{\vert T\vert}(v_t^{R+} + v_t^{R-})\le \Gamma^R
\end{array}
\right\}
\end{align}
\end{subequations}

Next, we conduct a comparative performance study on the DIU set \eqref{DIU} and the proposed DDU formulation \eqref{uncertain:P2}. Robust scheduling with DIU set \eqref{DIU} is solved by C\&CG algorithm. The first case is set up with $\Gamma^R=0$ in \eqref{DIU} and $E^{R+}=E^{R-}=0$ in \eqref{uncertain:P2}, respectively. It turns out that they obtain the same result that the net cost of VPP is 298495.42\$. We assume it to be the objective value of the nominal problem where no reserve deployment uncertainty exists and the price of robustness is calculated based on this value in the following cases. The second to the fourth cases study the impact of $\Gamma^R,E^{R+},E^{R-}$ on the price of robustness for VPP respectively and the results are depicted in Fig.\ref{fig20}-\ref{fig25}. As can be observed, price of robustness rises with an increasing uncertainty budget, but exhibits a different rate of change in DIU and DDU formulations. From the view of price of robustness, DIU set with $\Gamma^R=16$ is approximately a counterpart of the DDU set with decisions $E^{R+}=4800,E^{R-}=5175.95$. Recall that the optimal $E^{R+},E^{R-}$ are 5289.72 and 5175.95 respectively according to the results in subsection \ref{results}, indicating that a higher level of reserve budget is tolerable for VPP, considering the reserve energy revenue it provides. The proposed DDU formulation has the capability and incentive to strike the balance between robustness and profitability, by optimizing over the reserve budget rather than regarding it as a fixed parameter.

\begin{figure}[!htb]
	\centering
	\includegraphics[width=0.9\linewidth]{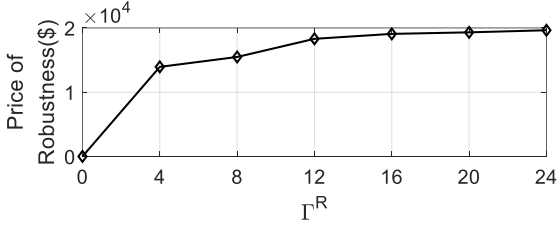}
	\caption{Impact of robustness parameter $\Gamma^R$ on the price of robustness.}
	\label{fig20}
\end{figure}
\begin{figure}[!htb]
	\centering
	\includegraphics[width=0.9\linewidth]{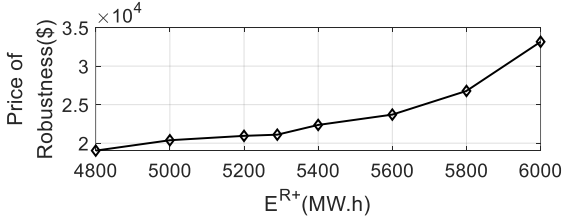}
	\caption{Impact of $E^{R+}$ on the price of robustness when fixing $E^{E-}$ to 5175.95.}
	\label{fig23}
\end{figure}
\begin{figure}[!htb]
	\centering
	\includegraphics[width=0.9\linewidth]{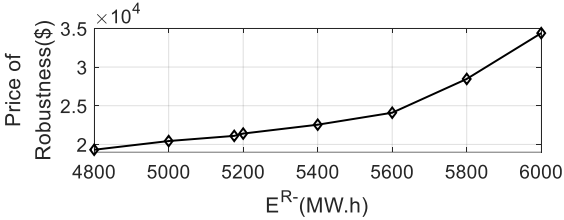}
	\caption{Impact of $E^{R-}$ on the price of robustness when fixing $E^{R+}$ to 5289.72.}
	\label{fig25}
\end{figure}

\subsubsection{Comparison between C\&CG algorithm and the proposed algorithm}
To emphasize the necessity of the proposed algorithm for decision-dependent robust optimization problem, we apply the widely used C\&CG algorithm to the problem and show how the C\&CG algorithm fails to guarantee solution optimality when the uncertainty is decision-dependent. Evolutions of objective value with the number of iterations in both algorithms are depicted in Fig.\ref{fig27}. The C\&CG algorithm converges fast, after 4 iteration rounds. However, the net cost of VPP derived by C\&CG algorithm is much greater than its optimal value. This is because, in the C\&CG algorithm, feasibility cut is directly generated by the worst-case uncertainty, ignoring that the uncertainty set is varying with decisions. The worst-case uncertainty realization in previous iterations may no more lie in the uncertainty set under some other decisions. Thus the feasibility cut of C\&CG algorithm may ruin the optimality of the solution, leading to over-conservative results.
 
\begin{figure}[ht]
	\centering
	\includegraphics[width=0.9\linewidth]{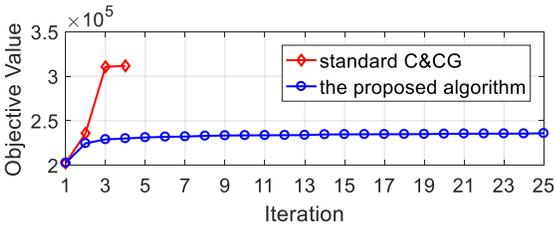}
	\caption{Comparision between the standard C\&CG algorithm and the proposed algorithm.}
	\label{fig27}
\end{figure}

% comparison between standard C\&CG algorithm and the proposed algorithm in VPP robust scheduling 

\section{Conclusion}
A novel stochastic adaptive robust optimization approach dealing with decision-dependent uncertainties is proposed in this paper for the DA scheduling strategies of a VPP participating in energy-reserve market. Consideration of the decision dependency of uncertain reserve deployment requests on VPP's offering in reserve market reduced the robustness of robust scheduling. The VPP determined the optimal level of robustness, striking a balance between the price of robustness and its profitability in the market. The proposed modified Benders decomposition algorithm obtained the optimum scheduling result under decision-dependent uncertainties, covering the shortage of standard C\&CG algorithm. Future works will address the consideration of better computational efficiency and a wider variety of decision dependent uncertainty sets.

\appendix
We start the proof of Theorem \ref{theorem:0} with the following lemmas.

\begin{lemma}
	\label{lemma:1}
	Let $f^*$ denote the optimal objective value of ARO-DDU \eqref{compact}. $k$ denotes the iteration round of Algorithm.\ref{alg}. Define $f^k:=f(x^k,y^{0,k})$. Then,
	\begin{itemize}
		\item [(a)] $f^k$ is monotonously non-decreasing with respect to $k$. 
		\item [(b)] For any $k\in\mathbb{Z}^+$, $f^k\le f^*$.
		\item [(c)] For any $k\in\mathbb{Z}^+$, if $R(x^k)=0$, $f^k\ge f^*$.
		\item [(d)] For any $k\in\mathbb{Z}^+$ and any $j\in [k]$, $\pi_j^*\in {\rm vert}(\Pi)$ where the set ${\rm vert}(\Pi)$ represents all the vertices of the polytope $\Pi$.
		\item [(e)] For any $k\in\mathbb{Z}^+$, $\forall j_1,j_2\in [k]$ and $j_1\neq j_2$, ${\pi}_{j_1}^*\neq {\pi}_{j_2}^*$.
	\end{itemize}
\end{lemma}

\begin{proof}
	\textit{Proof of Lemma \ref{lemma:1}(a):} Recall that $f^k$ is the optimal objective to the minimization master problem at iteration $k$. Since more and more constraints which are called feasibility cuts are appended to the minimization master problem \eqref{ineq:alg:5} during iterations, thus $f^k$ must be monotonously non-decreasing with respect to $k$.
	
	\textit{Proof of Lemma \ref{lemma:1}(b):} Recall the equivalent formulation of problem \eqref{compact} in \eqref{equivalent}, thus the master problem \eqref{ineq:alg:5} is always a relaxation to the minimization ARO-DDU problem \eqref{compact} for any $k\in\mathbb{Z}^+$. Thus $f^k\le f^*$ for any $k\in\mathbb{Z}^+$.
	
	\textit{Proof of Lemma \ref{lemma:1}(c):} Recall the definition of $R(x)$ in \eqref{ineq:alg:2}, $R(x^k)=0$ implies that $x^k$ satisfies constraint \eqref{equivalent:3}. Moreover, since $x^k$ is the solution to master problem \eqref{ineq:alg:5}, constrain \eqref{equivalent:2} (i.e., constraint \eqref{ineq:alg:5:2}) is met with $x^k$. Thus $x^k$ is a feasible solution to the minimization problem \eqref{equivalent}, indicating that $f^k\ge f^*$.

	\textit{Proof of Lemma \ref{lemma:1}(d):}
	Lemma \ref{lemma:1}(d) can be easily verified by noting that the optimal solution of bi-linear programming with polyhedron feasible set can be achieved at one of the vertices of the polytopes\cite{1976A}. Specific illustration is given as follows. For given $x^k$, since $(w_k^*,\pi_k^*)$ is the optimal solution to $R(x^k)$,
	\begin{align}
	(w_k^*,\pi_k^*)\in\arg\max_{w\in\mathcal{W}(x^k)}\left\{\max_{\pi\in \Pi}(b-Ax^k-Cw)^{\mathsf{T}}\pi\right\}
	\end{align}
	Then there must be $\pi_k^*\in\arg \max_{\pi\in \Pi}(b-Ax^k-Cw_k^*)^{\mathsf{T}}\pi$. 
	By noting that the unique optimal solution of linear programming must be found at one of its vertices, we have $\pi^*\in {\rm vert}(\Pi)$. 
	
	\textit{Proof of Lemma \ref{lemma:1}(e):}
	Suppose for the sake of contradiction that there exists $j_1,j_2\in [k]$ and $j_1\neq j_2$ such that ${\pi}^*_{j_1}={\pi}^*_{j_2}$. Without loss of generality we assume that $j_1<j_2$, and thus $j_1\le j_2-1$ since $j_1,j_2\in\mathbb{Z}^+$. Suppose ${\pi}_{j_2}^*$ is the optimal solution to $R({ x}^{j_2})$, there must be $R\left({ x}^{j_2}\right)>0$, implying that 
	\begin{align}
	\label{ineq:proof:3:1}
	\max_{{w}\in\mathcal{W}\left({ x}^{j_2}\right)}{{\pi}_{j_2}^*}^{\mathsf{T}}({ b}-A{x}^{j_2}-C{w})>0.
	\end{align}
	Since ${\pi}^*_{j_1}={\pi}^*_{j_2}$, we have
	\begin{align}
	\label{ineq:proof:3:2}
	\max_{{w}\in\mathcal{W}({ x}^{j_2})}{{\pi}_{j_1}^*}^{\mathsf{T}}({b}-A{ x}^{j_2}-C{w})>0.
	\end{align}
	Recall that ${x}^{j_2}$ is the optimal solution to the master problem with the following feasibility cuts
	\begin{align}
	\label{ineq:proof:3:3}
	0 \ge {{\pi}_j^*}^{\mathsf{T}}\left({b}-A{x}-C{ w}\right),\forall {w}\in\mathcal{W}({x}),j\in[j_2-1].
	\end{align}
	Since $j_1\le j_2-1$, there must be
	\begin{align}
	\label{ineq:proof:3:4}
	0 \ge {{\pi}_{j_1}^*}^{\mathsf{T}}\left({b}-A{ x}^{j_2}-C{w}\right),\forall { w}\in\mathcal{W}({x}^{j_2})
	\end{align}
	which contradicts with \eqref{ineq:proof:3:2}.
\end{proof}

Now we give the proof of Theorem \ref{theorem:0}.
\begin{proof}
	According to Lemma \ref{lemma:1}(a)-(b), $f^k$ is monotonously non-decreasing with respect to $k$ with an upper bound $f^*$. Combining Lemma \ref{lemma:1}(b) and (c), when the Algorithm.\ref{alg} terminates with $R(x^k)=0$, we have $f^k=f^*$, verifying the optimality of the solution.
	
	Next we illustrate that the Algorithm.\ref{alg} terminates within finite rounds of iterations. The number of vertexes of $\Pi$, denoted by $p$, is finite and no vertex of $\Pi$ can be appended twice to the master problem in Algorithm.\ref{alg} according to Lemma \ref{lemma:1}(d)-(e). Thus the Algorithm.\ref{alg} terminates within $\mathcal{O}(p)$ iterations.
\end{proof}

\ifCLASSOPTIONcaptionsoff
  \newpage
\fi
\bibliographystyle{IEEEtran}
\bibliography{IEEEabrv,mybib}

\end{document}